\documentclass[11pt, onecolumn]{article}
\usepackage[twoside, left=1in, right=1in, top=1.0in, bottom=1.0in]{geometry}

\usepackage[ntheorem]{empheq}
\usepackage[amsmath,framed,thmmarks,amsthm]{ntheorem}
\usepackage{algorithmicx}
\usepackage{algpseudocode} 
\usepackage{algorithm}

\usepackage{amssymb} 
\usepackage[subnum]{cases} 
\usepackage{graphicx}
\usepackage[svgnames]{xcolor}
\usepackage{enumerate}
\usepackage{caption}
\usepackage{subcaption}
\usepackage{cite}
\usepackage{hyperref}

\usepackage{aliascnt}
\usepackage[capitalize]{cleveref}

\newcommand{\CondProb}[2]{\mathbb{P}\left[{#1}\,\middle\vert\,{#2}\right]}
\newcommand{\defeq}{\triangleq}
\newcommand{\ints}{\mathbb{Z}}
\newcommand{\Latency}{T}
\newcommand{\nchoosek}[2]{{{#1}\choose{#2}}}
\newcommand{\eg}{e.g.}
\newcommand{\cf}{cf.}
\newcommand{\CondE}[2]{\mathbb{E}\left[\left.\!{#1}\right|#2\right]}
\newcommand{\DLPolicy}{\pi_{\mathrm{DL}}}
\newcommand{\ie}{i.e.}
\newcommand{\Prob}[1]{\mathbb{P}\left[{#1}\right]}
\newcommand{\indicator}[1]{\mathbb{I}\left\{{#1}\right\}}
\newcommand{\nn}{\nonumber}
\newcommand{\fLatency}[1]{T\left({#1}\right)}
\newcommand{\cB}{\mathcal{B}}
\newcommand{\cA}{\mathcal{A}}
\newcommand{\pints}{\mathbb{Z^+}}
\newcommand{\fCloudCost}[1]{\CloudCost\left({#1}\right)}
\newcommand{\abs}[1]{\left|#1\right|}
\newcommand{\CloudCost}{C_\mathrm{loud}}
\newcommand{\EMTradeoff}{\MachineTime-\ExecutionTime{} tradeoff}
\newcommand{\cP}{\mathcal{P}}
\newcommand{\cW}{\mathcal{W}}
\newcommand{\cV}{\mathcal{V}}
\newcommand{\cU}{\mathcal{U}}
\newcommand{\addtikz}[1]{\includegraphics{tpdfs/#1}}
\newcommand{\SLPolicy}{\pi_{\mathrm{SL}}}
\newcommand{\E}[1]{\mathbb{E}\left[{#1}\right]}
\newcommand{\bb}{\mathbf{b}}
\newcommand{\List}[1]{\left[{#1}\right]}
\newcommand{\PLOTDIR}{plots}
\newcommand{\bt}{\mathbf{t}}
\newcommand{\iid}{{i.i.d.}}
\newcommand{\bw}{\mathbf{w}}
\newcommand{\bp}{\mathbf{p}}
\newcommand{\cR}{\mathcal{R}}
\newcommand{\cardinality}[1]{\left|#1\right|}
\newcommand{\nnreals}{\mathbb{R}_+}
\newcommand{\Set}[1]{\left\{{#1}\right\}}
\DeclareMathOperator*{\argmin}{arg\,min}
\newcommand{\etc}{etc.}
\newcommand{\posfunc}[1]{\left\vert{#1}\right\vert^+}
\newcommand{\beiid}{\stackrel{~\iid}{\sim}}
\newcommand{\SetDef}[2]{\left\{#1 \setst #2 \right\}}
\newcommand{\ExecutionTime}{$\E{T}$}
\newcommand{\MachineTime}{$\E{C}$}
\newcommand{\setst}{:}

\newtheorem{theorem}{Theorem}
\newtheorem{definition}{Definition}
\newtheorem{remark}{Remark}

\newaliascnt{axiom}{theorem}

\aliascntresetthe{axiom}

\newaliascnt{lemma}{theorem}
\newtheorem{lemma}[lemma]{Lemma}
\aliascntresetthe{lemma}

\newaliascnt{prop}{theorem}

\aliascntresetthe{prop}

\newaliascnt{corollary}{theorem}
\newtheorem{corollary}[corollary]{Corollary}
\aliascntresetthe{corollary}

\newaliascnt{conjecture}{theorem}

\aliascntresetthe{conjecture}

\newaliascnt{observation}{theorem}

\aliascntresetthe{observation}

\newaliascnt{algo}{theorem}

\aliascntresetthe{algo}

\newaliascnt{notation}{definition}

\aliascntresetthe{notation}

\crefname{equation}{}{}
\Crefname{equation}{}{}
\crefname{thm}{theorem}{theorems}
\Crefname{thm}{Theorem}{Theorems}
\crefname{app}{appendix}{appendices}
\Crefname{app}{Appendix}{Appendices}
\crefname{prop}{proposition}{propositions}
\Crefname{prop}{Proposition}{Propositions}
\crefname{figure}{fig.}{figures}
\Crefname{figure}{Fig.}{Figures}
\crefname{defn}{definition}{definitions}
\Crefname{defn}{Definition}{Definitions}
\crefname{fact}{fact}{facts}
\Crefname{fact}{Fact}{Facts}
\crefname{appendix}{appendix}{appendices}
\Crefname{appendix}{Appendix}{Appendices}
\crefname{algo}{algorithm}{algorithms}
\Crefname{algo}{Algorithm}{Algorithms}
\crefname{algorithm}{algorithm}{algorithms}
\Crefname{algorithm}{Algorithm}{Algorithms}
\crefname{conjecture}{conjecture}{conjectures}
\Crefname{conjecture}{Conjecture}{Conjectures}
\crefname{obs}{observation}{observations}
\Crefname{obs}{Observation}{Observations}

\title{Efficient Task Replication for Fast Response Times\\in Parallel Computation}

\author{
Da Wang, Gauri Joshi, Gregory Wornell
\\
~\\
\small
Signals, Information and Algorithms Laboratory~\\
\small
Research Laboratory of Electronics~\\
\small
Massachusetts Institute of Technology
}

\date{}

\newcommand{\REFNAME}{paper}

\begin{document}
\maketitle

\begin{abstract}
    One typical use case of large-scale distributed computing in data centers is to decompose a
computation job into many independent tasks and run them in parallel on different machines,
sometimes known as ``embarrassingly parallel'' computation.  For this type of computation, one
challenge is that the time to execute a task for each machine is inherently variable, and the
overall response time is constrained by the execution time of the slowest machine.  To address
this issue, system designers introduce task replication, which sends the same task to multiple
machines, and obtains result from the machine that finishes first. While task replication
reduces response time, it usually increases resource usage. In this work, we propose a
theoretical framework to analyze the trade-off between response time and resource usage. We
show that, while in general, there is a tension between response time and resource usage, there
exist scenarios where replicating tasks judiciously reduces completion time and resource
usage simultaneously.  Given the execution time distribution for machines, we investigate the
conditions for a scheduling policy to achieve optimal performance trade-off, and propose
efficient algorithms to search for optimal or near-optimal scheduling policies. Our analysis
gives insights on when and why replication helps, which can be used to guide scheduler design
in large-scale distributed computing systems.

\end{abstract}

\section{Introduction}
\label{sec:intro}
One of the typical scenarios in cloud computing is large scale computation in a data centers
with a large number of computers, which is pioneered by companies like Google with the support
from distributed computing frameworks such as MapReduce~\cite{dean_mapreduce:_2008} and
Percolator~\cite{peng_large-scale_2010} , and distributed storage system such as Google File
System~\cite{ghemawat_google_2003} and BigTable~\cite{chang_bigtable:_2008}. Another canonical
example is the Amazon Web Service, where computing nodes can be obtained in a pay-as-you-go
fashion to accomplish computation at a wide range of scales.

An important category of large scale computation in data center is called ``embarrassingly
parallel'' computation~\cite{wikipedia_embarrassingly_2013}, where the computation can be
easily separated into a number of parallel tasks, often due to no dependency (or communication)
between these parallel tasks.  For an embarrassingly parallel job, we send each of its task to a
separate machine, let each machine execute the task, and collect the results from each machine.
While appears to be simplistic, embarrassingly parallel computation happens (either in part or
in whole) in many non-trivial applications, such as the ``Map'' stage of MapReduce, genetic
algorithms, the tree growth step of random forest, and so on.  In addition, embarrassingly
parallel computation is a feature that algorithm designers seek due to its ease of
implementation, in optimization~\cite{boyd_distributed_2011} and MCMC
simulation~\cite{neiswanger_asymptotically_2013}.

For an embarrassingly parallel job, the completion time is determined by the slowest computing
node, as one needs to wait for all parallel tasks to finish. However, machine response time in
data centers are inherently variable due to factors such as co-hosting, virtualization, network
congestion, misconfiguration, \etc. Then as the computing scale increases, it is increasingly
likely that the slowest machine is going to drag down the job completion time significantly.
For example, \cite[Table 1]{dean_tail_2013} shows that while the 99\%-percentile finishing time
for each task is 10ms, the 99\%-percentile finishing time for the slowest task in a large
computation job could take up to 140ms. Indeed, as pointed out by
practitioners~\cite{dean_achieving_2012,dean_tail_2013}, \emph{curbing latency variability}
is key to building responsive applications at Google. 

System designers have come up with a variety of techniques to curbing latency
variability~\cite{dean_tail_2013}, one of them being task replication, \ie, sending the same task
to more the one machines and take the result of whichever finishes first. While this approach
of replicating tasks is known to be effective in reducing task
completion time, it incurs additional resource usage as more machine running time is needed.
On one extreme, replicating the same task many times reduces the completion time variation
significantly, but results in high resource usage. On the other extreme, no replication is
incurs no additional resource usage, but often leads to long task completion time.  In
this paper, we aim to understand this trade-off between completion time and resource usage, and
based on our analysis, propose scheduling algorithms that are efficient in terms of both
completion time and resource usage.

In particular, we introduce a class of stylized yet realistic system models that enable us to
analyze this trade-off analytically or numerically. Our analysis reveals when and why task
replication works, and provides intuition for scheduling algorithm designs in practical
distributed computing systems.

\subsection{Related prior work}
The idea of replicating tasks is recognized by system designers for
parallel commutating~\cite{ghare_improving_2005,cirne_efficacy_2007}, and is first adopted in cloud
computing via the ``backup tasks'' in MapReduce~\cite{dean_mapreduce:_2008}.  A line of
system work
\cite{zaharia_improving_2008,ananthanarayanan_reining_2010,dean_achieving_2012,ananthanarayanan_effective_2013}
further develop this idea to handle various performance variability issues in data centers.

While task replication has been adopted in practice, to the best of our
knowledge, it has not been mathematically analyzed. By contrast, 
for scheduling without task replication, 
there exists a considerable amount of work on \emph{stochastic scheduling}, \ie,
scheduling problems with stochastic processing time~(\cf{}
\cite{pinedo_scheduling:_2012} and references therein).

Finally, some other work also investigate using replication or redundancy to reduce latency
in other contexts such as data
transfer~\cite{joshi_coding_2012,joshi_delay-storage_2013,vulimiri_low_2013,shah_when_2013}.

\subsection{Our contribution}
To the best of our knowledge, we establish the first theoretical analysis of efficient
task replication, by proposing the system model and relevant performance measures. 
Our findings show that:
\begin{enumerate}
    \item While in general there is a trade-off between completion time and resource usage,
        there exists scenarios where replicating tasks helps reduce both completion
        time and resource usage. 
    \item Given the machine execution time distribution, and the number of available machines, 
      we show that the search space for the optimal scheduling policy lies can be reduced to 
      a discrete and finite set of policies. 
    \item When the machine execution time follows a bimodal distribution, we find the optimal
        single-task scheduling policy for two special cases---the two machine case, and the
        single fork case. 
    \item We propose heuristic algorithms to choose the scheduling policy for both single-task
        and multi-task cases. These algorithms can achieve close to optimal performance with low
        computational cost. 
    \item We show that when scheduling multiple tasks, it is useful to take the interaction
        of completion times among different tasks into account, \ie, scheduling each task
        independently can be strictly suboptimal.
\end{enumerate}

\subsection{Organization of the paper}
The rest of the paper is organized as follows. In \Cref{sec:formulation} we define the
notation and describe the scheduling system model. 
Then we provide a motivating example in \Cref{sec:motivating_eg}.
In \Cref{sec:single_task_results} and \Cref{sec:multi_task_results} we provide a summary of our
results on single-task and multi-task scheduling respectively. The detailed analysis for both
single-task and multi-task scheduling are provided in \Cref{sec:proofs_single_task,sec:proofs_multi_task}. We
conclude the paper with brief discussion in \Cref{sec:conclu}.

\section{Model and notation} 
\label{sec:formulation}
\subsection{Notation}
\label{sec:notation}
We introduce here the notation that will be employed throughout the paper.
We use $\nnreals$ to denote all the non-negative real numbers, and $\pints$ all positive
integers. We use $[n]$ to denote all positive integers no larger than $n$, \ie, the set
$\Set{1, 2, \ldots, n}$.

We use bold font to represent a vector, such as $\bt = [t_1, \ldots, t_m]$. We use 
$[a, \bt]$ and $[\bt, a]$ to denote the vector resulting from appending an element
$a$ to the head and tail respectively of the vector $\bt$. 
For any number $x$, we denote 
\begin{equation*}
    \posfunc{x} = \max\Set{0, x}.
\end{equation*}

We use lower case letters (\eg\ $x$) to denote a particular value of the
corresponding random variable denoted in capital letters (\eg\ $X$).
We use ``w.p.'' as a shorthand for ``with probability''.

\subsection{System model}
\label{subsec:sys_model}
We consider the problem of executing a collection of $n$ \emph{embarrassingly parallel tasks} in
a data center.  We assume the execution time of each task on a 
machine in the data center is \iid{} with distribution $F_X$.

A \emph{scheduling policy} requests machines, and assigns tasks to different
machines, possibly at different time instants. More specifically, 
a scheduling policy $\pi$ is specified by a list of tuples 
\begin{equation*}
\pi \defeq \List{(a, t_{i,j}), a \in \cA, i \in [n], t_{i,j} \in \nnreals, j \in \pints}
,
\end{equation*}
where $\cA$ is the set of scheduling actions, $i$ is the task of interest, and $t_{i,j}$ is the
start time for the $j$-th copy of task $i$.

We assume set of scheduling actions $\cA$ contains the following two actions:
\begin{enumerate}
    \item \textsf{AddTaskToMachine}: the scheduler requests a machine to use from the pool of
        available machines and sends a task to run on the machine.
    \item \textsf{TerminateTask}: the scheduler shuts down all machines that are running a task. 
\end{enumerate}

We assume instantaneous \emph{machine completion feedback} is available from each machine notifying
the scheduler when it finishes executing the assigned task. 
This is a reasonable approximation as in general the task execution time is much longer than the
network transmission delay in a data center.

With machine completion feedback information, assuming we always terminate all copies of task $i$
when the earliest copy of task $i$ finishes, the performance of a scheduling policy is
determined by the times for action \textsf{AddTaskToMachine} only. Therefore, we simplify and say a
scheduling policy is specified by the time that it launches machines, \ie, 
\begin{equation*}
\pi = \List{t_{i,j}, i \in [n], t_{i,j} \in \nnreals, j \in \pints}.
\end{equation*}

Let $X_{i,j}$ be the running time of the $j$-th copy of task $i$ if it is not terminated, then
$X_{i,j} \beiid F_X$, and the completion time $T(\pi, i)$ for task $i$ satisfies
\begin{equation*}
    T_i \defeq T_i(\pi) \defeq \min_j (t_{i,j} + X_{i,j})
    .
\end{equation*}

\paragraph{Execution time distribution}~\\
While in practice a task can finish at any time and hence the execution time random variable
should be continuous valued, throughout the paper we model the execution time $X$ as a discrete
random variable, which corresponds to a probability mass function $P_X$, \ie, 
\begin{align}
    X &= \alpha_i \text{ w.p. } p_i, \quad 1 \leq i \leq l ,
    \\ 
    \text{or, } P_X(\alpha_i) &= p_i.
    \label{eq:pmf_def}
\end{align}
where $p_i \in [0, 1]$ and $\sum_{i=1}^l p_i = 1$.
We make this modeling choice for the following reasons:
\begin{enumerate}
    \item In practice we need to estimate the execution time distribution based on 
        log files or traces, and any estimation is more conveniently conducted assuming a
        discrete distribution. For example, a simple estimation could be a histogram of the
        past execution time spans with certain bin size (\eg, 10 seconds).
    \item We can use PMF to derive the upper bound of the performance by constructing the PMF
        in the following way: we set $P_X(\alpha_i) = p_i$ if 
        $p_i$-fraction of the execution time of a single task is within $\alpha_i$. 
    \item Depending on the state $i$ of a machine, its completion time could fall into 
        a range around $\alpha_i$, where state $i$ has probability $p_i$. 
\end{enumerate}
In particular, we often assumes $P_X$ is a \emph{bimodal distribution}, which corresponds to non-zero
probability at two time spans, \ie, 
\begin{equation}
    \label{eq:bimodal_def}
    X = \begin{cases}
        \alpha_1 & \text{ w.p. } p_1
        \\
        \alpha_2 & \text{ w.p. } p_2 = 1 - p_1
    \end{cases}
        .
\end{equation}
This modeling choice is motivated by the phenomenon of
``stragglers''~\cite{dean_mapreduce:_2008}, which indicates the majority of machines in the
data centers finish execution in the normal time span, while a small fraction of the machines
takes exceedingly long to complete execution due to malfunctioning of one or multiple part of
the data center, such as network congestion, software bugs, bad disk, \etc. 
In the bimodal distribution \Cref{eq:bimodal_def}, $\alpha_1$ can be viewed as the
time span that a normal machine takes to execute a task, and $\alpha_2$ the time span
that a straggler takes. 
Indeed, this is
observed from real system data, as pointed out by 
\cite[Observation 3]{chen_analysis_2010}, which states task durations are bimodal,
with different task types having different task duration distributions.

\paragraph{Static and dynamic launching}~\\
A scheduling policy corresponds to a choice of the vector of starting times
$\List{t_{i,j}, i \in [n], j \in \pints}$, and 
depends on when the starting times are chosen, we categorize a policy into static launching or
dynamic launching. 

A static launching policy chooses the starting time vector 
$$
\List{t_{i,j}, i \in [n], t_{i,j} \in \nnreals, j \in \pints}
$$
at time $0$ and does not change it afterwards.  
A dynamic launching policy would change the starting time vector during the execution process
by taking into the machine completion status into account.
    While the static launching policy takes less information into account and hence could be
    potentially less efficient, it allows more time for resource provisioning as we know the
    entire starting vector at $t=0$, hence may be of interest in certain applications or data
    centers.

\subsection{Performance metrics}
We evaluate the performance of a scheduling policy $\pi$ by the following two performance metrics:
\begin{itemize}
    \item completion time $T(\pi)$: the time that at least one copy of every task finishes
        running;
    \item machine time $C(\pi)$: sum of the amount of running time for all machines.
\end{itemize}
In addition to being a measure of resource usage, the machine time $C(\pi)$ can be viewed
    as a proxy for cost of using a public cloud, such as Amazon Web Service (AWS), which
    charges user per hour of machines used. 

For the $j$-th machine that runs task $i$, if the machine starting time $t_{i,j} \leq T_i$, then
it is run for $T_i - t_{i, j}$ seconds, otherwise it is not used at all. Hence, the running
time for this machine is $\posfunc{T_i-t_{i,j}}$. Therefore, 
\begin{align}
    T(\pi) &\defeq \max_i T_i(\pi)
    \label{eq:completion_time}
    \\
    C(\pi) &\defeq \frac{1}{n} \sum_{i=1}^n \sum_j \posfunc{ T_i(\pi) - t_{i,j} }
    \label{eq:machine_time}
    .
\end{align}
\Cref{fig:multi_task_scheduling} contains an example that illustrates a scheduling policy and
its corresponding completion time and node time.  Given two tasks, we launch task 1 at node 1
and 2 at $t_{1,1} = 0$ and $t_{1,2} = 2$ respectively, and task 2 at node 1 and 2 at $t_{2,1} =
0$ and $t_{2,2} = 5$ respectively.  The running time $X_{1,1} = 8$ and $X_{1,2} = 7$, and since
node 1 finishes the task first at time $t=8$, $T_1 = 8$ and node 2 is terminated before it
finishes executing.  Similarly, node 3 is terminated as node 4 finishes task 2 first at time
$T_2 = 10$. 
    The machine time for each machine is their actual running time, which are 8, 6, 10
    and 5 respectively, and hence the total machine time is the sum 29, while completion
    time is $T = \max\Set{T_1, T_2} = 10$.
\begin{figure}[h!]
    \begin{center}
        \addtikz{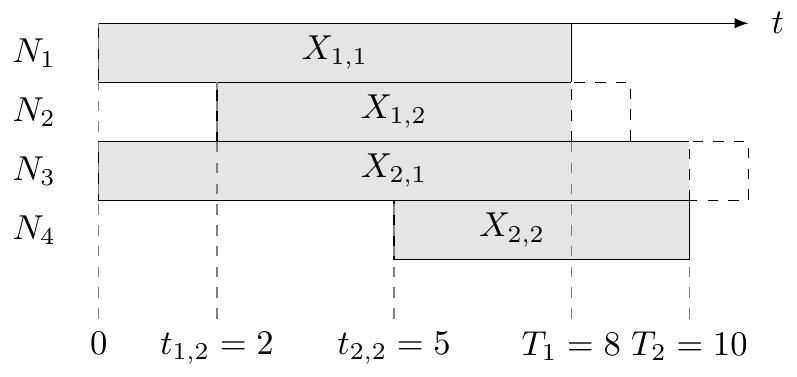}
    \end{center}
    \caption{Example illustrating a scheduling policy and its performance, where
        $\Set{t_{i,j}}$ are the starting times for tasks, and $\Set{X_{i,j}}$ are the running
        time for tasks. The machine time is $C=29$ and the completion time is $T=10$.
    }
    \label{fig:multi_task_scheduling}
\end{figure}

\paragraph{Cost function}~\\
Intuitively, while introducing task replication reduces $T$, it may incur additional resource usage
and hence increase $C$. In this work we investigate trade-off between these two quantities. 
In particular, we define the following cost function: 
\begin{equation}
    \label{eq:cost_func}
    J_\lambda(\pi) = \lambda \E{T(\pi)} + (1-\lambda) \E{C(\pi)}, 
\end{equation}
where $0 \leq \lambda \leq 1$ reflects the relative importance of completion time.

\begin{remark}
    $\lambda$ can be used to take cost of completion time and cost of computing resource into
    account. $\lambda=1$ and $\lambda=0$ correspond to the case of caring about completion
    time only and machine time only, respectively.
\end{remark}

\subsection{Optimal and suboptimal policies}
The introduction of cost function $J_\lambda(\cdot)$ allows us to compare policies directly, 
and we define optimal and suboptimal policies.
\begin{definition}[Optimal and suboptimal policies]
\label{def:optimal_and_suboptimal_policy}
\label{def:suboptimal_policy}
Given $\lambda$, then the corresponding \emph{optimal} scheduling policy $\pi^*$ is
\begin{equation*}
    \pi^* = \argmin_{\pi} J_\lambda(\pi)
    .
\end{equation*}
\end{definition}
\begin{remark}
    Note that there may exist policies that are neither optimal nor suboptimal.
\end{remark}

However, the search space for optimal policy is non-trivial, as the cost function is not
non-convex, and the search space is large, because we can launch any number of machines at any
time before $\alpha_l$. 

For the rest of the paper, we tackle the optimization problem by narrowing down the search
space, solving for special yet important cases, and proposing heuristic algorithms.

\section{Motivating example}
\label{sec:motivating_eg}
In this section we consider the following example, which shows in certain scenarios, task
replication reduces both $\E{T}$ and $\E{C}$, even for a single task!

Let the execution time $X$ satisfies
\begin{equation*}
    X = \begin{cases}
        2 & \text{w.p. } 0.9
        \\
        7 & \text{w.p. } 0.1
    \end{cases}
    .
\end{equation*}

\begin{figure}[h!]
    \centering
    \begin{subfigure}[b]{0.49\textwidth}
        \centering
        \addtikz{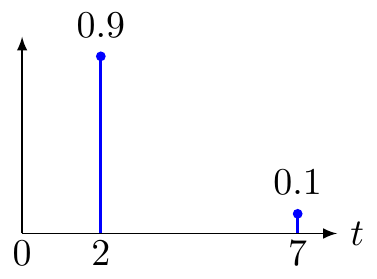}
        \caption{$P_{T}$ without replication}
        \label{fig:bimodal_exec_dist}
    \end{subfigure}
    \begin{subfigure}[b]{0.49\textwidth}
        \centering
        \addtikz{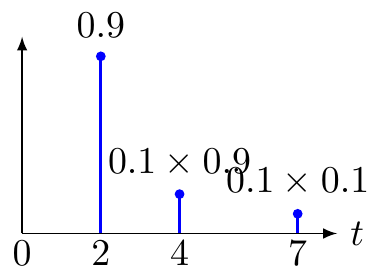}
        \caption{$P_T$ with replication at $t=2$}
        \label{fig:bimodal_exec_dist_r}
    \end{subfigure}
    \caption{Execution time distribution}
    \label{fig:bimodal_exec_dists}
\end{figure}

If we launch one task and wait for its completion, then the completion time distribution is illustrated in
\Cref{fig:bimodal_exec_dist}, and
\begin{align}
    \Latency &= 2 \times 0.9 + 7 \times 0.1 = 2.5
    \\
    \CloudCost &= \Latency = 2.5
    .
\end{align}
If we launch a task at time $t_1 = 0$ and then launch a replicated task at time $t_2 = 2$ if the first one has not finished running by then, 
we have the completion time distribution in
\Cref{fig:bimodal_exec_dist_r}, and in this case,
\begin{align*}
    \Latency &= 2 \times 0.9 + 4 \times 0.09 + 7 \times 0.01 = 2.23
    \\
    \CloudCost &= 2 \times 0.9 + (4+2) \times 0.09 + (7 + 5) \times 0.01 = 2.46
    .
\end{align*}
As we see here, introducing replication actually reduces both expected cost and
expected execution time!

\section{Single-task scheduling}
\label{sec:single_task_results}
In this section we present our results regarding the optimal scheduling for a single
task. While this seems simplistic, it is practically useful if we cannot divide a job into multiple parallel tasks. In
addition, it is impossible to scheduling multiple tasks optimally if we do not even understand how to schedule a single
task optimally.

We postpone all proofs to \Cref{sec:proofs_single_task}.

We first note that in a single-task scheduling scenario, we can represent a scheduling policy by its starting time
vectors, \ie, 
\begin{equation*}
    \pi = \bt = [t_1, t_2, \ldots, t_m] 
    ,
\end{equation*}
where $t_j$ is the time that the task starts on machine $j$.

\begin{remark}
    Note that the starting time vector $[t_1, \ldots, t_m]$ is equivalent to 
    $[t_1, t_2, \ldots,$ $t_m, \alpha_l, \ldots, \alpha_l]$ as tasks scheduled to start at
    $\alpha_l$ will never be launched. We use the two representations interchangeably in this
    \REFNAME. 
\end{remark}

The performance metrics, completion time $\Latency$ and cost $\CloudCost$, can now be expressed as
\begin{align}
    \Latency &= \min_{1 \leq j \leq m} t_j + X_j ,
    \label{eq:single_task_T_def}
    \\
    \CloudCost &= \sum_{j=1}^{m} \posfunc{T- t_j} ,
    \label{eq:single_task_C_def}
\end{align}
where $X_j \beiid P_X$.

We then show in \Cref{thm:static_optimal} that in single-task scheduling, dynamic launching and static
launching policies are equivalent in the sense that they achieve the same \EMTradeoff{}.
\begin{theorem}
    For single task scheduling, 
    the static launching policy achieves the same \EMTradeoff{} region as the dynamic launching policy.
    \label{thm:static_optimal}
\end{theorem}
\begin{remark}
    The above result does not hold for scheduling multiple tasks in general, as the dynamic
    launching policy can take different actions depending on if any other tasks are finished.
\end{remark}

Therefore, for the single-task scenario, we can focus on the static launching policy without any loss of generality.

\subsection{General execution time distribution}
\label{sec:single_task_general}
Given the machine execution time distribution $P_X$ and a starting time vector $\bt = [t_1, \ldots, t_m]$, we first show
an important property of $\E{T(\bt)}$ and $\E{C(\bt)}$ in \Cref{thm:linearity}.
\begin{theorem}
    \label{thm:linearity}
    $\E{T(\bt)}$ and $\E{C(\bt)}$ are piecewise linear functions of $\bt$. 
\end{theorem}
A further refinement of \Cref{thm:linearity} results \Cref{thm:corner_points_optimal}, which
indicates the optimal starting time vector $\bt \in [0, \alpha_l]^m$ is located in a finite
set, which is composed by a constrained integer combination of the support of $P_X$.
\begin{theorem}
    \label{thm:corner_points_optimal}
    The starting time vector $\bt = [t_1, \ldots, t_m]$ that minimizes $J_\lambda$ satisfies that
    \begin{equation}
        t^*_j \in \cV_m,
        \label{eq:lattice_pt_cond}
    \end{equation}
    where $\cV_m$ is a finite set such that
    \begin{align}
        \label{eq:corner_pts}
        \cV_m \defeq \SetDef{v}
        {
            v = \sum_{j=1}^l \alpha_{j} w_{j},
            0 \leq v \leq \alpha_l,
            \sum_{j=1}^l \abs{w_j} \leq m,
            w_{j} \in \ints
        }
        .
    \end{align}
\end{theorem}
\Cref{thm:corner_points_optimal} directly leads to \Cref{coro:simple_corner_point_case}.
\begin{corollary}
    \label{coro:simple_corner_point_case}
    If PMF $P_X$ satisfies that $\alpha_j = k_j \alpha, 1 \leq j \leq l, k_j \in \pints$, then
    the optimal starting time vector $\bt^*$ satisfies
    \begin{equation*}
        t_j \in \cV_m \subset \Set{0, \alpha, 2\alpha, \ldots, \alpha_l = k_m\alpha}
        ,
    \end{equation*}
    where $\cardinality{\cV_m} \leq k_m + 1$.
\end{corollary}
Given \Cref{thm:corner_points_optimal}, we can calculate the $\E{T}$ and $\E{C}$ for all
starting time vectors that satisfy \Cref{eq:lattice_pt_cond}, then discard suboptimal ones,
leading to the \EMTradeoff{} as shown in \Cref{fig:eg_L3}, which are plotted for the following
two execution times:
\begin{align}
    X &= \begin{cases}
        4  & \text{ w.p. } 0.6
        \\
        8  & \text{ w.p. } 0.3
        \\
        20 & \text{ w.p. } 0.1
    \end{cases}
    ,
    \label{eq:X_L3_eg}
    \\
    X' &= \begin{cases}
        6  & \text{ w.p. } 0.8
        \\
        20 & \text{ w.p. } 0.2
    \end{cases}
    \label{eq:X_eg}
    .
\end{align}
\begin{figure}[h!]
    \begin{center}
    \begin{subfigure}[b]{0.49\textwidth}
        \includegraphics[width=0.99\textwidth]{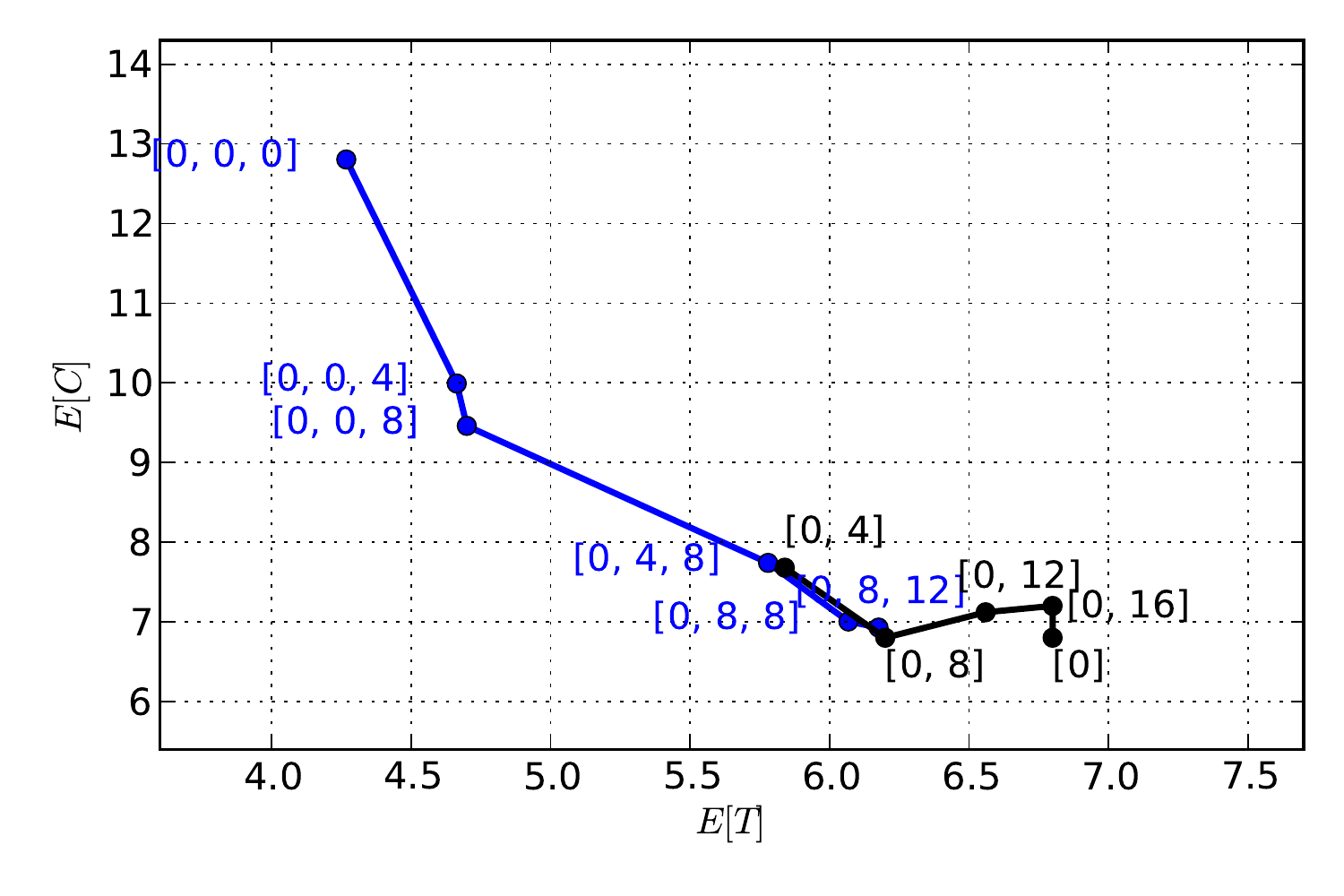}
        \caption{Execution time $X$ in \Cref{eq:X_L3_eg}}
    \end{subfigure}
    \begin{subfigure}[b]{0.49\textwidth}
        \includegraphics[width=0.99\textwidth]{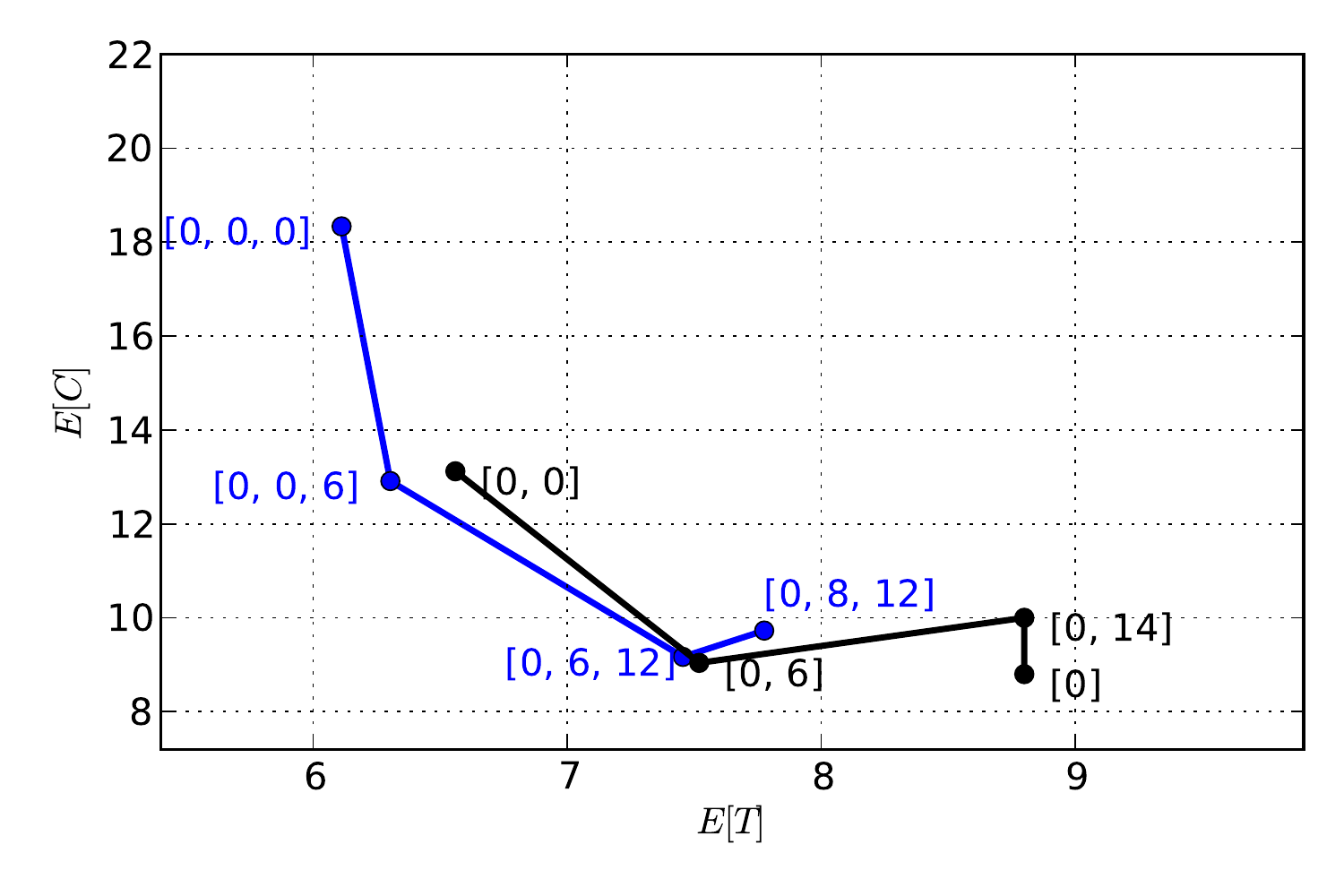}
        \caption{Execution time $X'$ in \Cref{eq:X_eg}}
    \end{subfigure}
    \end{center}
    \caption{Examples of the \EMTradeoff{} with $m=3$ machines.
        The labels for each point is the corresponding starting time vector, and the region is
        defined by two piecewise linear segments, which are colored blue and black
        respectively.
    }
    \label{fig:eg_L3}
\end{figure}

Furthermore, we show that the optimal choice of the $(i+1)$-th element of the starting time vector is
dependent on the starting times before it, \ie, $t_1, t_2, \ldots, t_{i}$, via
\Cref{thm:corner_pt_optimal}.
In particular, the optimal value belongs to a set $\cU$ that we called \emph{corner points} and
define in \Cref{def:corner_points}.
\begin{definition}[Corner points]
    \label{def:corner_points}
    Given $\bt = [t_1, t_2, \ldots, t_i]$, let 
    \begin{align*}
        \cU_1 &\defeq \Set{0, \alpha_1, \ldots, \alpha_l}, 
        \\
        \cU_{i+1}(t_1, \ldots, t_i)
        &\defeq 
        \bigcup_{u \in \cU_i(t_1, \ldots, t_{i-1})} 
        \bigg\{u + t_i - b \alpha_j :
            \\
            &\;\; \qquad 0 \leq u + t_i - b \alpha_j \leq \alpha_l,
            \\
            &\;\; \qquad 1\leq j \leq l, b \in \Set{0,1} \bigg\}, \quad i \geq 1
        ,
    \end{align*}
    and we called $\cU_{i+1}$ the corner points given $\bt$.
\end{definition}
\begin{theorem}
    \label{thm:corner_pt_optimal}
    Given $\bt = [t_1, t_2, \ldots, t_i]$ and the corner points $\cU_{i+1}(\bt)$, then  
    the optimal scheduling policy with $i+1$ machines
    \begin{equation*}
        \bt' = [t_1, t_2, \ldots, t_i, t_{i+1}]
    \end{equation*}
    satisfies
    \begin{equation*}
        t_{i+1} \in \cU_{i+1}.
    \end{equation*}
\end{theorem}

Finally, we have the following simple observation that, again, help to reduce the search space
of scheduling policy.
\begin{lemma}
\label{lem:last_corner_sub_optimal}
Starting a machine at any time $\alpha_l - \alpha_1 \leq t \leq \alpha_l$ is suboptimal.
\end{lemma}

\subsubsection{Heuristic policy search algorithm}
While \Cref{thm:corner_points_optimal} reduces the search space of the optimal scheduling
policy, there could still be exponentially many policies to evaluate.  In this section we
introduce a heuristic single-task scheduling algorithm in \Cref{algo:single_task_myopic} that has
much lower complexity.

As shown in \Cref{algo:single_task_myopic}, this heuristic algorithm builds the starting time
vector $[t_1, t_2, \cdots t_m]$ iteratively, with the constraint that $t_i$'s are in
non-decreasing order.  Given a starting time vector $[t_1, \cdots t_i]$, this algorithm
compares the policies $[t_1, \cdots t_i, t_{i+1}]$ where $t_{i+1}$ can be one of first $k$
corner points in $\cU(t_1, \ldots, t_i)$, and choose the policy $t_{i+1}$ that leads to the
minimum cost.  As we increase $k$, the algorithm compares a larger space of policies and hence
achieves a lower cost as illustrated by the example in \Cref{fig:cmp_L3}, for the execution
time defined in \Cref{eq:X_L3_eg}.
The example also
demonstrates that a small $k$ may be sufficient to achieve near-optimal cost. 
\begin{algorithm}
\begin{algorithmic}
    \State Initialize $t_1 = 0$ and $\bt = [t_1]$
    	\For{$i = 2, \ldots m$}
            \State $U^+(\bt) \gets$ sorted elements of $\cU(\bt)$ which are $\geq t_{i-1}$
            \State $\pi_0 \gets [\bt, \alpha_l]$, policy that keeps the machine unused
            \For{$j = 1, \ldots, k$}
                    \State  $\pi_j \gets [\bt, U^+(\bt)[j]]$  
     		\EndFor
            \State $j^* \gets \argmin_{j \in {0, 1, \cdots k}}  J_{\lambda}(\pi_j)$
         	\State $t_i \gets U^+(\bt)[j]$ and $\bt \gets [\bt, t_i]$
    	\EndFor
\end{algorithmic}
\caption{$k$-step heuristic algorithm for single-task scheduling}
\label{algo:single_task_myopic}
\end{algorithm}
\begin{figure}[h!]
    \begin{center}
        \includegraphics[width=0.5\textwidth]{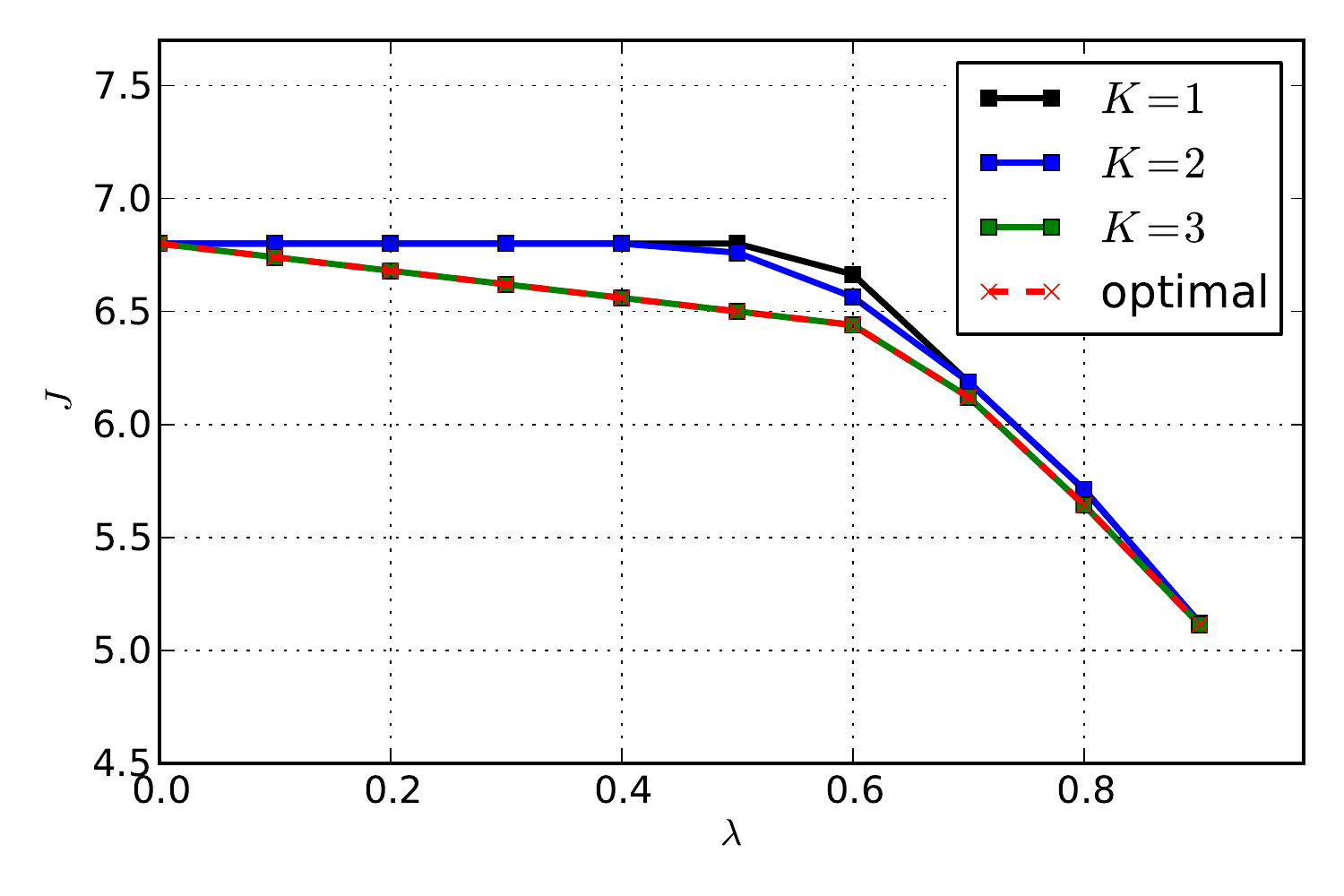}
    \end{center}
    \caption{Comparison between the heuristic search and optimal scheduling policy for
        execution time in \Cref{eq:X_L3_eg}.
    }
    \label{fig:cmp_L3}
\end{figure}

\subsection{Bimodal execution time distribution}
While results in \Cref{sec:single_task_general} help characterize the \EMTradeoff{}
and find good scheduling policies, they provide little insight about when and why task
replication helps. For this, we analyze the special yet important case of bimodal execution time
distribution (\cf{} \Cref{eq:bimodal_def}).

In this section we present results for scheduling one task with two machines, which is the
simplest non-trivial example. The scheduling policy can be represented as the vector $\bt =
[t_1 = 0, t_2]$, and we provide a complete characterization of the \EMTradeoff{} in 
\Cref{fig:bimodel_eg_tradeoff}, leading to \Cref{thm:bimodal_2m_tradeoff}.
\begin{figure}[tb!]
    \begin{center}
        \addtikz{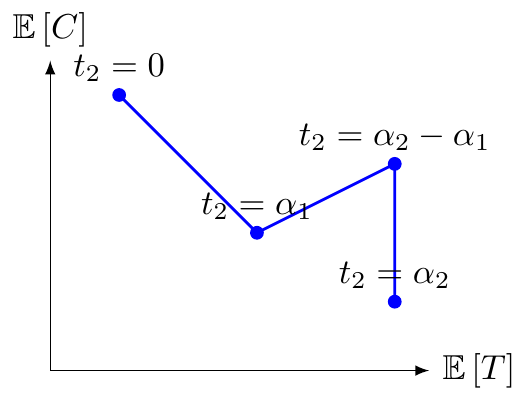}
    \end{center}
    \caption{The $\E{T}$-$\E{C}$ tradeoff for bimodal execution with two machines, which corresponds to
        starting time vector $\bt=[t_1=0, t_2]$.}
    \label{fig:bimodel_eg_tradeoff}
\end{figure}
\begin{theorem}
    \label{thm:bimodal_2m_tradeoff}
    Given $P_X$ is a bimodal distribution and we have at most two machines, 
    the optimal policy $\bt = [t_1 = 0, t_2]$ satisfies $t_2 \in \Set{0, \alpha_1, \alpha_2}$.
\end{theorem}
In \Cref{thm:bimodal_2_mach_opt_policy}, we provide further insights by showing the
suboptimality (\cf{} \Cref{def:optimal_and_suboptimal_policy}) of certain scheduling policies as 
the execution time distribution $P_X$ varies, which is characterized by
the ratio of its fast and slow response time, $\alpha_1/\alpha_2$, and the probability that it
finishes at its fast response time, $p_1$. 
\begin{theorem}
\label{thm:bimodal_2_mach_opt_policy}
Given the bimodal execution time and two machines, 
\begin{enumerate}[(a)]
    \item $[0, \alpha_2-\alpha_1]$ is always suboptimal;
    \item $[0, \alpha_1]$ is suboptimal 
        if $\frac{\alpha_1}{\alpha_2} >\frac{p_1}{1+p_1}$;
    \item $[0, \alpha_2]$ is suboptimal 
        if $\frac{\alpha_1}{\alpha_2} < \frac{2p_1-1}{4p_1-1}$;
\end{enumerate} 
Given $\lambda$, we can find the optimal policy by comparing the ratio
$\frac{1-\lambda}{\lambda}$ to the thresholds, 
\begin{align}
\tau_1 &= \frac{ \alpha_1 p_1(3-2p_1) + \alpha_2(1-p_1)(1-2p_1)}{(\alpha_2-\alpha_1)(1-p_1)p_1} \\
\tau_2 &= \frac{1+ 2p_1(1-p_1)}{p_1(1-p_1)} \\
\tau_3 &= \frac{\alpha_1(4p_1-1) + \alpha_2 (1-2p_1)}{\alpha_2 - 2\alpha_1) p_1}
\end{align}
\begin{enumerate}[(a)]
\setcounter{enumi}{3}
\item 
    If $\frac{\alpha_1}{\alpha_2} >\frac{p_1}{1+p_1}$, then
    policy $[0,\alpha_2]$ is optimal if $\frac{1-\lambda}{\lambda}  \leq \tau_1$, and $[0,0]$ is optimal otherwise. 
\item 
    If $\frac{2p_1-1}{4p_1-1} \leq \frac{\alpha_1}{\alpha_2} \leq \frac{p_1}{1+p_1}$, then
    policy $[0, \alpha_1]$ is optimal if $ \tau_3 <\frac{1-\lambda}{\lambda} \leq
    \tau_2$, policy $[0,\alpha_2]$ is optimal if $\frac{1-\lambda}{\lambda}  \leq \tau_3$, and
    $[0,0]$ is optimal otherwise. 
\item If $\frac{\alpha_1}{\alpha_2} < \frac{2p_1-1}{4p_1-1}$, then
    policy $[0, \alpha_1]$ is optimal if $\frac{1-\lambda}{\lambda} \leq \tau_2$, and $[0,0]$ is optimal otherwise.
\end{enumerate}
\end{theorem}
\Cref{thm:bimodal_2_mach_opt_policy} is summarized in \Cref{fig:bimodal_2_mach_opt_policy_region}. 
\begin{figure}[tb!]
    \begin{center}
        \addtikz{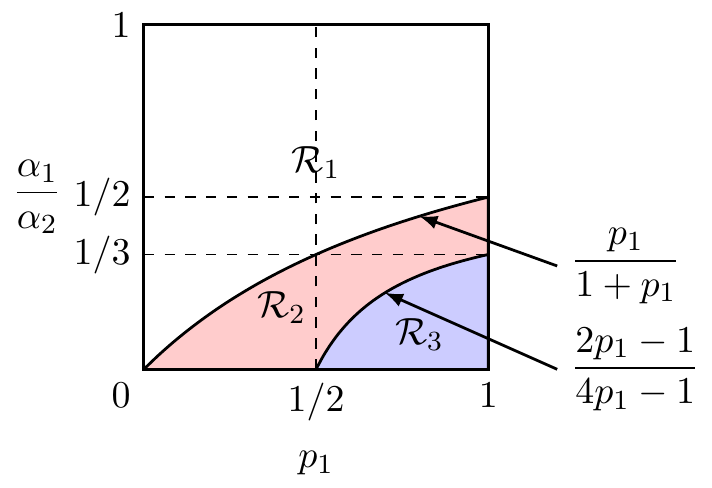}
    \end{center}
    \caption{Bimodal two machine. $\cR_1$ is the range of parameters that $\bt=[0,\alpha_1]$ is
        strictly suboptimal, $\cR_3$ is the range $\bt=[0, \alpha_2]$ is strictly suboptimal,
        which means no task replication is strictly suboptimal.}
    \label{fig:bimodal_2_mach_opt_policy_region} 
\end{figure}

\section{Multi-task scheduling}
\label{sec:multi_task_results}
In this section we investigate the scheduling of multiple tasks. We first show that it is crucial to take the
interaction of different tasks into account in \Cref{thm:separation_suboptimal}, then extend our algorithm in
\Cref{algo:single_task_myopic} for multi-task scheduling.
All proofs are postponed to \Cref{sec:proofs_multi_task}.

\begin{theorem}[Separation is suboptimal]
    \label{thm:separation_suboptimal}
    Given $m$ tasks, applying the optimal one-task scheduling policy for each of them individually
    is suboptimal.
\end{theorem}

Given the complexity of searching for optimal scheduling policy in the single-task case, we
again aim to search for scheduling policy via a heuristic algorithm. In particular, we aim to
find a good static policy that takes the interaction among tasks into account. To achieve this,
we apply \Cref{algo:single_task_myopic}, but using the cost function for the multi-task case,
where $T$ and $C$ are defined in \Cref{eq:completion_time} and \Cref{eq:machine_time} respectively.
This search procedure produces a starting time vector $\bt = [t_1, t_2, \ldots, t_m]$, and at
each time $t_i$, we launch an additional copy for each of the unfinished task.

\Cref{fig:multi_task_eg1} shows an example for the execution time in \Cref{eq:X_L3_eg}. The
scheduling policy with replication reduces $J$, especially when $\lambda$ is large. We also see 
that as the number of tasks $n$ increases, the cost $J$ increases as the
impact of the slowest task gets more severe. 
Again, introducing replication mitigates this degradation.

Results in \Cref{fig:multi_task_eg1} indicate that when $\lambda$ is not too big, it may be
beneficial to introduce replication at multiple time instants, as in this case, we are more
concerned with cost $\CloudCost$ and hence introducing replication gradually is preferred. By
contrast, when $\lambda$ is close to 1, a good scheduling policy should introduce replication
early to cut down completion time as early as possible.
\begin{figure}[h!]
    \begin{center}
    \begin{subfigure}[b]{0.49\textwidth}
        \includegraphics[width=0.97\textwidth]{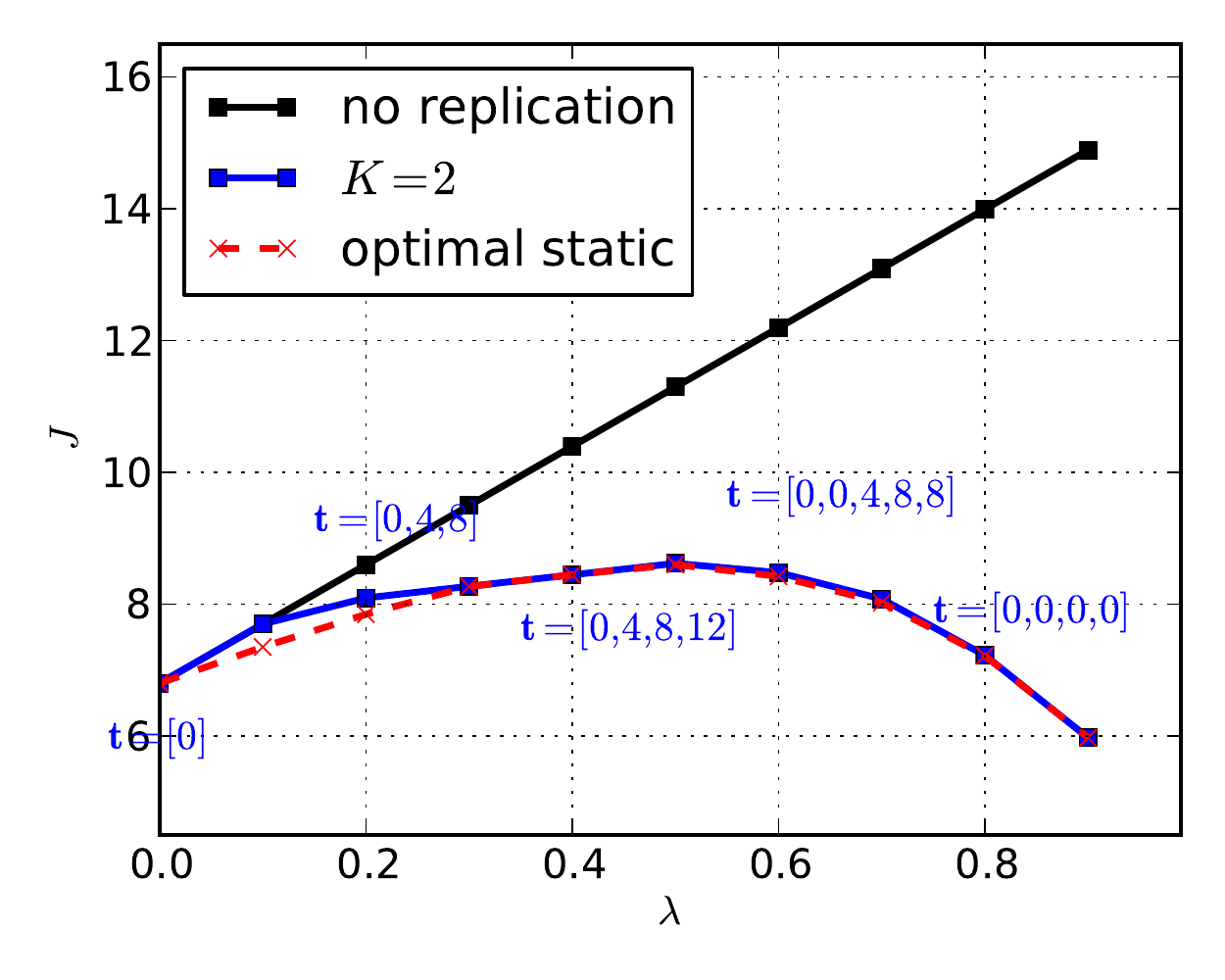}
        \caption{$N=10$ tasks}
    \end{subfigure}
    \begin{subfigure}[b]{0.49\textwidth}
        \includegraphics[width=0.97\textwidth]{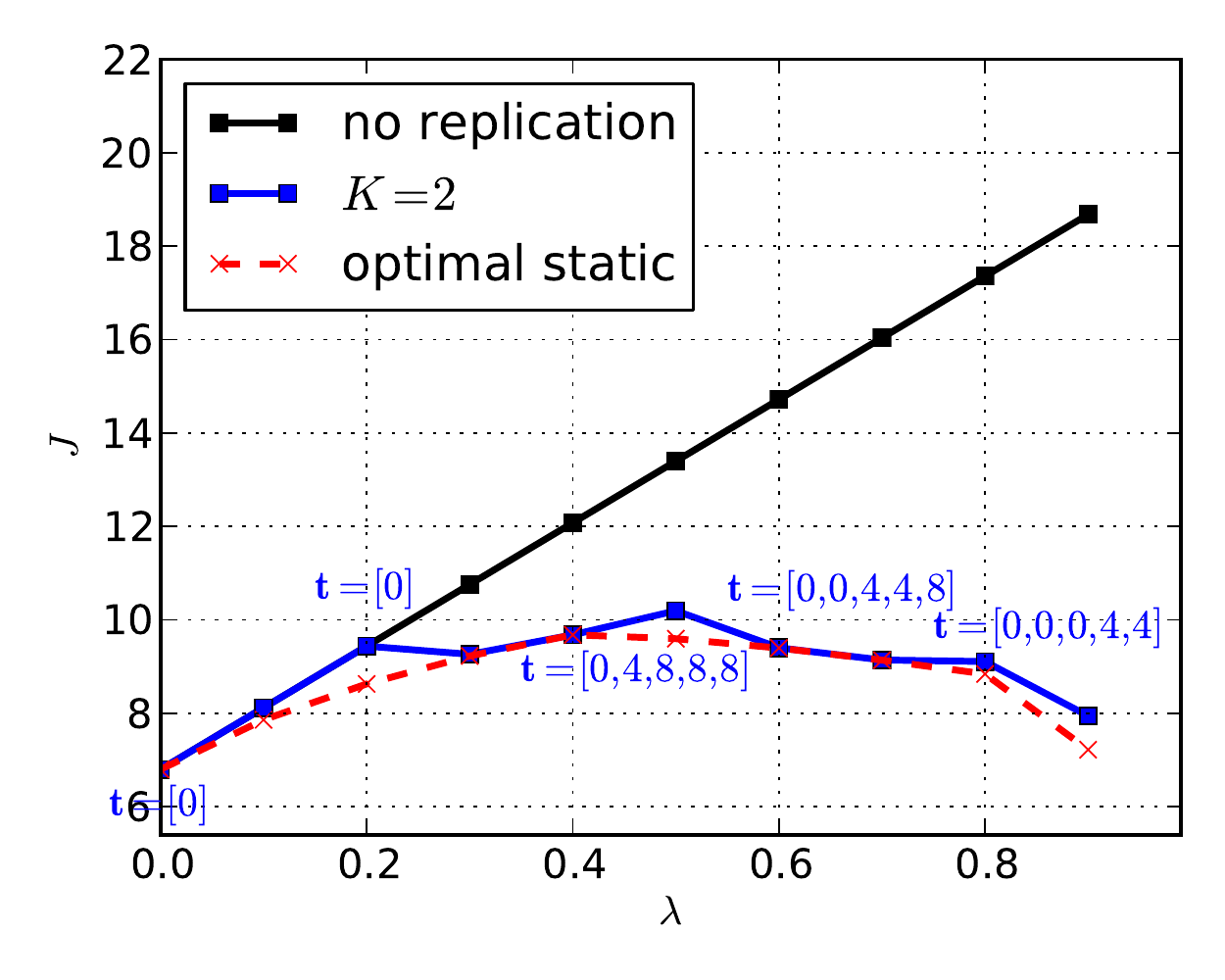}
        \caption{$N=100$ tasks}
    \end{subfigure}
    \end{center}
    \caption{Performance of scheduling policy based on heuristic search for
        execution time in \Cref{eq:X_L3_eg}. $k$ is the parameter in
        \Cref{algo:single_task_myopic}. 
        The starting time vector for $k=2$ and $\lambda = 0, 0.2, 0.4, 0.6$ and $0.8$ are
        labeled in the plots.
    }
    \label{fig:multi_task_eg1}
\end{figure}

\begin{remark}
    Our proposed policy via searching algorithm is static in nature. One may extend it by
    running the searching algorithm at each time instant.  For example, at time $t=0$ we obtain
    the starting time vector $t^{(0)}$.  At time $t=\alpha_1$ we can re-run the search
    algorithm given the number of unfinished tasks and obtain an updated starting time vector
    $t^{(1)}$, \etc. This policy is dynamic in nature and is likely to achieve better
    performance than the static policy. 
\end{remark}

\section{Proofs for single task scheduling}
\label{sec:proofs_single_task}
In this section we present the detailed analysis for the problem of optimal scheduling for a
single task.

\subsection{Proofs for \Cref{thm:static_optimal}}
\begin{proof}[Proof of \Cref{thm:static_optimal}]
    With node completion feedback, at any time $t$, a dynamic launching policy can make launching decision based on
    whether any running node finishes execution by then, and it is obvious that it should launch new copies only if no
    node has finished executing the task.

    Given a dynamic launching policy $\DLPolicy$, we can construct a static policy $\SLPolicy$ by letting its starting
    time vector be the starting time emitted from the dynamic policy under the condition that no machine finishes
    execution until $\alpha_l$. Now we claim $\DLPolicy$ and $\SLPolicy$ achieves the same $\E{T}$ and $\E{C}$, because
    it is not difficult to see that for a given realization of machine execution time, both policies will launch and
    terminate machines at the same time.
\end{proof}

\subsection{Proofs for \Cref{thm:linearity} and \Cref{thm:corner_points_optimal}}
In this section we show that the \EMTradeoff{} curve is always piecewise linear, with
the vertices of the piecewise linear curve corresponding to starting time vector that
satisfies certain properties.

We first define the possible finishing time 
\begin{equation*}
    w_{i, j} \defeq t_i + \alpha_j, 1 \leq i \leq m, 1 \leq j \leq l
\end{equation*}
and the set of all possible finishing times
\begin{align*}
    \cW \defeq \Set{
        w_{i, j}, 1 \leq i \leq m, 1 \leq j \leq l
    }
    .
\end{align*}
Let $k = \cardinality{\cW}$, 
we denote the sorted version of $\cW$ as 
$\bw = [w_{\sigma_\bt(1)}, w_{\sigma_\bt(2)}, \ldots, w_{\sigma_\bt(k)}]$ such that
\begin{align*}
    w_{\sigma_\bt(1)} \leq w_{\sigma_\bt(2)} \leq \ldots \leq w_{\sigma_\bt(k)}
    ,
\end{align*}
where $\sigma_\bt(k)$ maps the rank of the finishing time $k$ to a tuple $(i_k,j_k)$. 

Note that 
\begin{align*}
    T = \min_{1\leq i \leq m} X_i + t_i
    ,
\end{align*}
and $T \in \cW$, 
we define the event 
\begin{align*}
    \cA_{k_1, k_2} \defeq \Set{
        \min_{1 \leq i \leq m, 1 \leq j \leq l}^*\Set{t_i + X_j} = t_{k_1} + \alpha_{k_2} 
    }
    ,
\end{align*}
where $\min^*$ indicates we always choose the smallest $(k_1, k_2)$ (by lexicographic order in
$k_1$ and $k_2$) so that all
the events $\Set{\cA_{k_1, k_2}, 1 \leq k_1 \leq m, 1 \leq k_2 \leq l}$ are disjoint.

Therefore, 
\begin{align}
    \E{T} 
    &= \sum_{k_1, k_2} \CondE{T}{\cA_{k_1, k_2}} \Prob{\cA_{k_1, k_2}}
    \nn
    \\
    &= \sum_{k_1, k_2} (t_{k_1} + \alpha_{k_2}) \Prob{\cA_{k_1, k_2}}
    \label{eq:cond_ET}
    \\
    \E{C} 
    &= \sum_{j=1}^m \E{C_j}
    \nn
    \\
    &= \sum_{j=1}^m \sum_{k_1, k_2} \posfunc{ t_{k_1} + \alpha_{k_2} - t_j} \Prob{\cA_{k_1, k_2}}
    \label{eq:cond_EC}
\end{align}
To analyze \Cref{eq:cond_ET} and \Cref{eq:cond_EC}, 
we first show that the relative ordering of elements in $\cW$ determines 
$\Prob{\cA_{k_1, k_2}}$. 
\begin{lemma}
    \label{lemma:prob_ordering}
    $\Prob{\cA_{k_1, k_2}}$ is independent of $\Set{\alpha_j, 1 \leq l \leq l}$ given the
    relative ordering of elements in $\cW$,
    \ie,
    \begin{equation}
        \label{eq:cond_prob_determined_by_ordering}
        \CondProb{\cA_{k_1, k_2}}{\sigma_\bt} = f(\sigma_\bt, k_1, k_2, p_1, \ldots, p_l)
        ,
    \end{equation}
    where $f$ is some function.
\end{lemma}
\begin{proof}
    $\cA_{k_1, k_2}$ indicates machine $k_1$ is the first machine that finishes execution, and it
    finishes execution after running for $\alpha_{k_2}$. 
    Define $k \defeq \sigma_\bt^{-1}(k_1, k_2)$, \ie, 
    \begin{equation*}
        t_{k_1} + \alpha_{k_2} = w_{\sigma_\bt(k)},
    \end{equation*}
    then
    \begin{align}
        \CondProb{\cA_{k_1, k_2}}{\sigma_\bt} 
        &= 
        \CondProb{ \bigcap_{j \neq k_1} \Set{t_j + X_j} > w_{k_1, k_2} }{\sigma_\bt}
        \nn
        \\
        &= 
        \prod_{j \neq k_1}
        \CondProb{ t_j + X_j > w_{k_1, k_2} }{\sigma_\bt}
        \label{eq:P_A_k}
        .
    \end{align}
    Define 
    \begin{equation*}
        \cP_j \defeq \SetDef{p}{\sigma_\bt(i) = (j, p), i > k}
        ,
    \end{equation*}
    which is uniquely determined by $\sigma_\bt$ and $k$, then for any $i \neq k_1$, 
    \begin{equation}
        q_j 
        \defeq 
        \CondProb{ t_j + X_j > w_{k_1, k_2} }{\sigma_\bt}
        =
        \sum_{p \in \cP_j} p
        \label{eq:q_j}
        ,
    \end{equation}
    which is a function of $k$, $\sigma_t$ and $\bp = [p_1, p_2, \ldots, p_m]$.

    Combing \Cref{eq:P_A_k} and \Cref{eq:q_j}, we have
    \Cref{eq:cond_prob_determined_by_ordering}.
\end{proof}

\begin{proof}[Proof of \Cref{thm:linearity}]
    We prove that there exists finitely many subspaces of $[0, \alpha_l]^m$ such that in each subspace, 
    $\E{T}$ and $\E{C}$ is a linear function in $\bt$, and thus they are piecewise linear in $\bt$ on
    $[0, \alpha_l]^m$.

    Define $\cB_1(\sigma) \defeq \SetDef{\bt}{\sigma_\bt = \sigma}$, and $(k_1, k_2) = \sigma(k)$, 
    then the set 
    \begin{align}
        \label{eq:B1_def}
        \cB_1(\sigma) &= \SetDef{\bt}{
            t_{1_1} + \alpha_{1_2} \leq t_{2_2} + \alpha_{2_2} \leq 
            \ldots 
            \leq t_{k_1} + \alpha_{k_2}
        }
    \end{align}
    is defined by $k-1$ inequalities, and each of this inequality partition the space 
    $[0, \alpha_l]^n$ into two subspaces. Therefore, $\cB_1(\sigma)$ is the intersection of $k-1$
    connected subspaces, resulting itself being a subspace of $[0, \alpha_l]^m$.
    And it is obvious that there are only finitely many such subspaces.
    Therefore, by \Cref{lemma:prob_ordering} and \Cref{eq:cond_ET}, 
    in each subspace $\cB_1(\sigma)$, $\E{T}$ is a linear function in $\bt$.

    Regarding $\E{C}$, we define
    \begin{align}
        &\cB_2(\bb = [b_{i,j}]_{1\leq i\leq m, 1 \leq j \leq k}, \sigma) \subset \cB_1(\sigma)
        \nn
        \\
        \defeq&
        \big\{\bt:
            \sigma_\bt = \sigma, \;
            \indicator{t_{k_1} + \alpha_{k_2} - t_i > 0} \in \Set{0,1},
        \nn
        \\
            &\qquad (k_1, k_2) = \sigma^{-1}(j), 1 \leq j \leq k
        \big\},
        \label{eq:B2_def}
    \end{align}
    where $\indicator{\cdot}$ is the indicator function. Similar to the argument above, 
    given $\sigma$ and $\bb$, $\cB_2(\bb, \sigma)$ 
    corresponds to a subspace of $[0, \alpha_l]^m$ and there are only finitely many such subspaces. 
    By \Cref{lemma:prob_ordering} and \Cref{eq:cond_EC}, 
    in each subspace $\cB_2(\bb, \sigma)$, $\E{C}$ is a linear function in $\bt$.

    Therefore, both $\E{T}$ and $\E{C}$ are piecewise linear functions of $\bt$ in 
    $[0, \alpha_l]^m$.
\end{proof}

\begin{proof}[Proof of \Cref{thm:corner_points_optimal}]
    By \Cref{thm:linearity} and the fact that $J_\lambda$ is a linear combination of $\E{T}$
    and $\E{C}$, the optimal $\bt^*$ that minimizes $J_\lambda(\bt)$ is at 
    the boundaries of two or more subspaces defined in \Cref{eq:B2_def}.
        
    Then by \Cref{eq:B1_def,eq:B2_def}, it is not hard to see that 
    for some $j_1, j_2, j_3, j_4$ and $l_1, l_2, l_3$, we have
    \begin{align*}
        t^*_{j_1} - t^*_{j_2} &= \alpha_{l_1} - \alpha_{l_2}
        \\
        t^*_{j_3} - t^*_{j_4} &= \alpha_{l_3}
        .
    \end{align*}
    Then it is not hard to see that given $m$, $\bt = [t_1, t_2, \ldots, t_m]$, 
    and without loss of generality, let $t_1 = 0$, 
    \begin{equation*}
        t^*_i \in \cV_m,
    \end{equation*}
    where $\cV_m$ is defined in \Cref{eq:corner_pts}, \ie, 
    \begin{equation*}
        \cV_m \defeq \SetDef{v}
        {
            v = \sum_{j=1}^l \alpha_{j} w_{j},
            0 \leq v \leq \alpha_l,
            \sum_{j=1}^l \abs{w_j} \leq m,
            w_{j} \in \ints
        }
        .
    \end{equation*}
    Note that an element in $\cV_m$ is uniquely determined by $\bw = [w_1, \ldots, w_l]$, and
    the number of possible $\bw$ is 
    \begin{equation*}
        2^l \nchoosek{m+l-1}{l-1}.
    \end{equation*}
    Therefore, 
    \begin{equation*}
        \cardinality{\cV_m} \leq 2^l \nchoosek{m+l-1}{l-1} \leq [2(m+l-1)]^{l},
    \end{equation*}
    which is finite given finite $m$ and $l$.
\end{proof}

\subsection{Proofs related to corner points}

\begin{proof}[Proof of \Cref{thm:corner_pt_optimal}]
    Let $U_{i+1} = [u_1, u_2, \ldots, u_{k_i}]$ be the sorted version of $\cU_{i+1}$, then 
    $\E{T(\bt')}$ and $\E{C(\bt')}$ are linear in $t_{i+1}$ over the each interval
    $[u_j, u_{j+1}], 1 \leq j \leq k_i-1$. 
    Therefore, the optimal $t_{i+1} \in \cU_{i+1}$.
\end{proof}

\subsection{Proof of \Cref{lem:last_corner_sub_optimal}}
\begin{proof}[Proof of \Cref{lem:last_corner_sub_optimal}]
Consider a set of $m$ machines on which we run the task according to the policy $\pi = [t_1,
t_2, \cdots, t_m]$. Without loss of generality, we assume $t_1 = 0$. 
If the starting time of a machine is $\alpha_l > t_j \geq
\alpha_l - \alpha_1$, the earliest time it can finish execution of the task is $t + \alpha_1$.
This time is greater than $\alpha_l$, the latest time at which the first machine started at
time $t_1 = 0$ finishes the task. Thus, starting the machine at time $t_j$ only adds to the
cost $\E{C}$, without reducing the completion time $\E{T}$. Hence, any starting time
$t_j \geq \alpha_l - \alpha_1$ should be replaced by $\alpha_l$, which corresponds to not using
that machine at all.  
\end{proof}

\subsection{Proof of \Cref{thm:bimodal_2m_tradeoff}}
\label{sec:proof_bimodal_2m_tradeoff}
\Cref{thm:bimodal_2m_tradeoff} follows directly from the following lemma.
\begin{lemma}
    \label{lemma:bimodal_2m_exprs}
    Given $P_X$ is a bimodal distribution and we have at most two machines, 
    the expected completion time and total cost satisfies that if 
    $t_2 + \alpha_1 < \alpha_2$,
        \begin{align*}
            \E{T} &= \alpha_1 (p_2 - p_1) p_1 + \alpha_2 p_2^2 + t_2 p_1 p_2
            ,
            \\
            \E{C} &= \begin{cases} 
                2\E{T} - t_2 (p_1^2 + p_2^2) & \text{ if } t_2 < \alpha_1
                \\
                2\E{T} - \alpha_1 p_1 - t_2 p_2 & \text{ if } t_2 \geq \alpha_1
            \end{cases}
        ;
        \end{align*}
    otherwise if $ t_2 + \alpha_ 1 \geq \alpha_2$,
        \begin{align*}
            \E{T} &= \alpha_1 p_1 + \alpha_2 p_2
            \\
            \E{T} &= 2\E{T} - \alpha_1 p_1 - t_2 p_2
        \end{align*}
\end{lemma}
\begin{proof}
    By \Cref{eq:single_task_C_def,eq:single_task_T_def}
    and calculation.
\end{proof}

\subsection{Proof of \Cref{thm:bimodal_2_mach_opt_policy}}
\label{sec:bimodal_2mach}

\begin{proof}[Proof of \Cref{thm:bimodal_2_mach_opt_policy}]

\begin{enumerate}[(a)] 
\item Follows from \Cref{lem:last_corner_sub_optimal}. 
\item If $\frac{\alpha_1}{\alpha_2} >\frac{1}{2}$ then by \Cref{lem:last_corner_sub_optimal} we know that if $[0,\alpha_1]$ is suboptimal. Now suppose $\frac{\alpha_1}{\alpha_2} >\frac{1}{2}$. We know that $ \pi_1 = [0, 0]$ and $ \pi_2 = [0, \alpha_2]$ are the two extreme ends of the $(\E{C}, \E{T})$ trade-off. If the line joining points $(\E{C(\pi_1)}, \E{T(\pi_1)})$ and $(\E{C(\pi_2)}, \E{T(\pi_2)})$ lies below $(\E{C(\pi)}, \E{T(\pi)})$, then $\pi = [0, \alpha_1]$ will be suboptimal. Comparing the slopes of the lines gives the condition $\frac{\alpha_1}{\alpha_2} > \frac{p_1}{1+p_1}$. 

\item Policy $\pi_2 = [0, \alpha_2]$ is suboptimal if it is dominated by either $\pi = [0, \alpha_1]$ or $\pi_1 = [0,
    0]$. Both $\pi$ and $\pi_1$ give lower expected execution time $\E{T}$ than $[0, \alpha_2]$. So if one of them has
    expected cost $\E{C}$ lower than $[0, \alpha_2]$, then it follows that $[0, \alpha_2]$ is dominated by that
    strategy. But the $\E{C}$ with starting time vector $[0, 0]$ is always greater than
    that of $[0, \alpha_1]$. Thus, checking if the expected machine $\E{C}$ with $[0, \alpha_1]$ is smaller than that
    for $[0, \alpha_2]$, gives the condition $\frac{\alpha_1}{\alpha_2} < \frac{2p_1-1}{4p_1-1}$ for suboptimality of
    $[0, \alpha_2]$.
\item, (e), (f) 
    For cost function $J = \lambda \E{T} + (1-\lambda) \E{C}$, the constant cost contour is a line with slope
    $-\frac{1-\lambda}{\lambda}$. As we increase $J$, the contour line shifts upward until it hits the $(\E{C}, \E{T})$
    trade-off. The point where it meets the $(\E{C}, \E{T})$ trade-off corresponds to the optimal policy. In $\cR_1$,
    policy $\pi_2 = [0,\alpha_2]$ is optimal if the slope of the line joining $(\E{C(\pi_1)}, \E{T(\pi_1)})$ and
    $(\E{C(\pi_2)}, \E{T(\pi_2)})$ is less than or equal to $-\frac{1-\lambda}{\lambda}$. We can simplify and show that
    the slope of the line is $-\tau_1$. The result follows from this. Similarly, the slope of the line joining $[0,0]$,
    $[0,\alpha_1]$ is $-\tau_2$, and that of the line joining $[0,\alpha_1$ and $[0,\alpha_2]$ is $-\tau_3$. Comparing
    the slope of the contour, $-\frac{1-\lambda}{\lambda}$ with these slopes gives the conditions of optimality for each
    of the policies.

\end{enumerate}
\end{proof}

\section{Proofs for multi-task scheduling}
\label{sec:proofs_multi_task}
\subsection{Proof of \Cref{thm:separation_suboptimal}}
\begin{proof}
    We prove the statement by showing an example that a scheduling policy that takes the
    interaction of task latencies into account (joint policy) is better than a scheduling
    each task independently (separate policy).

Suppose we have two tasks and $4$ machines. The service time distribution of 
each machine is bimodal, taking values $\alpha_1$ and $\alpha_2 > \alpha_1$ with
probability $p_1$ and $1-p_1$ respectively. Assume $2 \alpha_1  < \alpha_2$.

\emph{Separate Policy} \newline
Consider a policy $\pi_s$ where we choose the optimal scheduling policy separately for each
task.  We can follow the analysis of the bimodal 2-machine case in \Cref{sec:bimodal_2mach} as a
guideline to choose the optimal policy for each task. 

Suppose the policy $[ 0, \alpha_2]$ is optimal for a given cost function. For this to be true, the
parameters $\alpha_1$, $\alpha_2$ and $p_1$ need to satisfy,
\begin{align}
\frac{\alpha_1}{\alpha_2} > \frac{2p_1 - 1}{4 p_1 - 1}
\end{align} 

If we run each task on two machines using the policy $[0, \alpha_2]$, the expected completion time 
and cost are,
\begin{align*}
    \E{\fLatency{\pi_s}} &= p_1^2 \alpha_1 + (1-p_1^2) \alpha_2, \\
    \E{\fCloudCost{\pi_s}} &=2  p_1^2 \alpha_1 + 2 p_1 (1-  p_1) (\alpha_1 + \alpha_2) + 2 (1-p_1^2) \alpha_2 .
\end{align*}

\emph{Joint Policy} \newline
Consider a joint policy $\pi_d$ where we start with each task according to policy $[ 0, \alpha_2]$. If task 1 (task 2)
is served by its machine at time $\alpha_1$, we start the execution the task 2 (task 1) on an additional machine
at time $\alpha_1$. 

Using this joint policy the performance metrics are given by
\begin{align*}
\E{T(\pi_d)} 
&= p_1^2 \alpha_1 + 2p_1^2 (1-p_1) (2 \alpha_1) + (1-p_1)^2(2p_1 +1) \alpha_2,\\
\E{C(\pi_d)} 
&=p_1^2 (2\alpha_1) + 2 p_1^2 (1-  p_1) (3\alpha_1) + (1-p_1)^2(2p_1 +1) (2\alpha_2).
\end{align*}

We can show that for $2 \alpha_1 < \alpha_2$, $\E{T(\pi_d)} < \E{T(\pi_s)}$. Now let us find the condition
for $\E{C(\pi_d)} < \E{C(\pi_s)}$.
\begin{align*}
\E{C(\pi_d)} &< \E{C(\pi_s)} \\
\Rightarrow \qquad \frac{\alpha_1}{\alpha_2} &< \frac{2 p_1 - 1}{3 p_1 -1}
.
\end{align*}
Thus, the joint policy gives strictly lower cost $J_{\lambda} = \lambda \E{C} + (1-\lambda) \E{T}$ than the
separate policy for any $\lambda$ if
\begin{align}
\frac{2 p_1 - 1}{4 p_1 -1} < \frac{\alpha_1}{\alpha_2} < \frac{2 p_1 - 1}{3 p_1 -1}
.
\end{align}
\end{proof}

\section{Concluding Remarks}
\label{sec:conclu}

In this paper we present the first theoretical analysis of how to effective replicate tasks
such that we reduce completion time, with minimum use of extra computing resources. 

We show that for certain scenarios, task replication may in fact simultaneously reduce
both execution time and resource usage, and in general, it leads to a
better response time and computing resource usage trade-off.

Given a discrete approximation to the service time distribution, we characterize the optimal
trade-off between execution time and resource usage for the case of scheduling a single task. We
show the optimal scheduling policy is in a set of finite size.  We also present a
low-complexity heuristic algorithm to choose the scheduling policy that is close to optimal.
Further, we give insights into extending this analysis to the multi-task case.

Our work answers the questions on when and how task replication helps, and our results provide
guidance to scheduling design in data centers, such as the time to launch tasks and the number of time
we should replicate it.

This work can be extended in a few directions.  First, one can search for better scheduling
policies, especially for the multi-task case.  Second, in our work we assume the execution time
distribution is given or can be estimated, it may be of interest to develop an adaptive
scheduling policy that does not require such knowledge.  Third, it will be useful to estimate
the error due to approximating a continuous execution time distribution by a discrete execution
time distribution, either numerically or via simulation. Finally, one can take the effect of
queueing of requests at the machines into account and see how that impacts the system
performance.


\bigskip
\subsection*{Acknowledgement}
We thank Devavrat Shah for helpful discussions.
\bigskip

\bibliographystyle{abbrv}
\bibliography{dwabrv,reference,websites}

\end{document}